\documentclass[12pt]{article} 

\usepackage{geometry} 
\usepackage{fullpage}

\usepackage{graphicx} 

\usepackage{booktabs} 
\usepackage{array} 
\usepackage{paralist} 
\usepackage{verbatim} 
\usepackage{subcaption}
\usepackage{cite}
\usepackage{graphicx}
\usepackage{hyperref,  mathrsfs, bm}

\usepackage{url}
\usepackage{ragged2e}
\usepackage{multirow}

\graphicspath{{./}{./../graphs/}{./graphs/}}

\usepackage{amssymb,amsmath,amsthm,amscd}

\usepackage[mathscr]{eucal}

\renewcommand{\raggedright}{\leftskip=0pt \rightskip=0pt plus 0cm}
\DeclareMathOperator*{\argmin}{arg\,min}

\usepackage{epsfig}
\usepackage{graphicx}
\usepackage[usenames,dvipsnames]{color}

\usepackage{ogonek}

\usepackage{tikz}
\usepackage{pgf}

\usepackage{fancyhdr} 
\usepackage[nottoc,notlof,notlot]{tocbibind} 
\usepackage[titles,subfigure]{tocloft} 

\newtheorem{thm}{Theorem}
\newtheorem{cor}{Corollary}
\newtheorem{lem}{Lemma}
\newenvironment{remark}[1][Remark]{\begin{trivlist}
\item[\hskip \labelsep {\bfseries #1}]}{\end{trivlist}}

\title{ Fidelity-Commensurability Tradeoff in Joint Embedding of Disparate Dissimilarities}

\author{Sancar Adali\footnote{
Corresponding Author: Johns Hopkins University,
Department of Applied Mathematics and Statistics,
100 Whitehead Hall,
3400 North Charles Street,
Baltimore, MD 21218-2682 ; \mbox{sadali1@jhu.edu}.
Acknowledgements: This work was partially supported by National Security Science and Engineering Faculty Fellowship (NSSEFF),
Johns Hopkins University Human Language Technology Center of Excellence (JHU HLT COE), and the
XDATA program of the Defense Advanced Research Projects Agency (DARPA) administered through Air Force Research Laboratory
contract FA8750-12-2-0303, the NSF BRAIN Early Concept Grants for Exploratory Research (EAGER) award DBI-1451081,
and Acheson J. Duncan Fund for the Advancement of Research in Statistics.} \and Carey E. Priebe\footnote{Johns Hopkins University}
}

\begin{document}
\maketitle
\abstract{ In  various data settings, it is necessary to compare observations  from disparate data sources. We assume the  data is in the dissimilarity representation\cite{duin2005dissimilarity} and investigate  a joint embedding method\cite{JOFC} that results in a commensurate representation of disparate dissimilarities. We further assume that there are ``matched'' observations from different conditions which can be considered to be highly similar, for the sake of inference. The joint embedding results in the joint optimization of fidelity (preservation of within-condition dissimilarities) and commensurability (preservation of between-condition dissimilarities between matched observations). We show that the tradeoff between these two criteria can be made explicit using weighted raw stress as the objective function for multidimensional scaling. In our investigations, we use a weight parameter, $w$, to control the tradeoff, and choose  match detection as the inference task. Our results show  weights  that are optimal (with respect to the inference task) are different than equal weights for commensurability and fidelity and  the proposed weighted embedding scheme provides significant improvements in statistical power.
}

\section{Introduction\label{sec:intro}}
  We are interested in problems where the data sources are disparate and the inference task requires that  observations from  the different data sources  be judged to be similar or dissimilar. By ``disparate'',  we  mean that the observations from the two sources are inherently incomparable, either because the data are of different modalities (such as images and text), or there is  significant and unknown variability between  the data sources (such as psychometric data collected from different subjects).  Throughout this paper, we refer to the disparateness of the observations when we mention they are from  different ``conditions''.

	Consider a collection of  English Wikipedia articles  and the hyperlink graph based on the links between the same articles. Each article corresponds to a vertex in this directed graph\footnote{For simplicity, assume the collection of articles correspond to a connected graph.}. This correspondence illustrates our idea of ``matchedness". This example also illustrates the idea of  ``disparate'': there is no intuitive way to compare a text document and a vertex in a graph.
	We assume the training data consists of  a collection of matched data from the disparate data sources.

  The inference task we consider is match detection, i.e. deciding whether a new English article and  a new vertex in  the graph are matched. While a document can be compared   with other documents in the same language via a dissimilarity measure defined for documents, and vertices in the same graph can be compared via a dissimilarity measure defined for graph vertices, a direct comparison between  a document and a graph vertex is not possible.
		To facilitate  our approach to the inference task at hand, it is necessary   to derive a data representation  where the  observations from different conditions can be compared, i.e. the representation is commensurate.  
	We will use a finite-dimensional Euclidean space for  this commensurate representation, where standard  statistical inference tools can be used.

		As in the wikipedia document/graph example, it is possible that a feature representation of the data is not available or inference with such a representation is fraught with complications. This motivates our  dissimilarity-centric approach. For an excellent resource on the usage of dissimilarities in pattern recognition, we refer the reader to the P\k{e}kalska and Duin book \cite{duin2005dissimilarity}.

		Since we proceed to inference starting from a dissimilarity representation of the data, our methodology may be applicable to any scenario in which multiple dissimilarity measures are available.  Some illustrative examples include:  pairs of images and their descriptive captions,  photographs taken under different illumination conditions. In each case, we have an intuitive notion of matchedness: for photographs taken under different illumination conditions, ``matched'' means they are of the same person. For a collection of linked Wikipedia articles, the different conditions  are  the textual content and hyperlink graph structure, ``matched'' means a text document and  a vertex  in the graph corresponds to the same Wikipedia article.

    To quantify how suitable the commensurate representation is for subsequent inference, two error criteria can be defined: \emph{fidelity}, which refers to how well the available dissimilarities in a condition are preserved, and \emph{commensurability}, which refers to how well the dissimilarities between matched objects are preserved. These two concepts will be made more concrete in section \ref{sec:FidComm}.

The major question  addressed in this paper is whether, in the tradeoff between fidelity and commensurability, there is a ``sweet spot'': increases in fidelity (or commensurability) do not result in superior performance for  the inference task, due to the resulting commensurability (or fidelity) loss.

\section{Related Work \label{sec:RelatedWork}}

 Our problem is very similar to the 3-way multidimensional scaling problem where the dissimilarity data is a ($n \times n \times K$) tensor which represent pairwise dissimilarities between $n$ objects as measured in $K$ different conditions. However, whereas  our embedding approach  finds a separate configuration of points for each condition, 3-way MDS methods \cite{Carroll_Chang_1970,proxscal} find a \emph{single} configuration of $n$ points representing  each object (which is referred to as ``group space''), that is as consistent as possible with the dissimilarity data under different conditions. DISTATIS\cite{DISTATIS} accomplishes this goal by finding a compromise inner product matrix that is a weighted combination of the inner product matrices in different conditions. 3-way MDS methods such as \mbox{INDSCAL}\cite{Carroll_Chang_1970}, PROXSCAL\cite{proxscal} assume a common configuration in group space, from which the individual dissimilarity matrices are computed after being distorted by weight matrices.  In contrast,  when we adopt the  joint ``commensurate'' embedding approach,  the representation of the objects in the common space is never estimated.  In fact, some of our inference tasks including match detection make sense only if  we represent  each object under each condition as a distinct point. The violation of global assumptions the \mbox{INDSCAL} and PROXSCAL methods make about the group space and the weight matrices  might result in significant performance loss.  Our method puts weaker constraints on the mapping between different conditions and allows the degree of matchedness to vary for different pairs of observations. As such, it should be more robust to noise (weakly matched pairs) in the training data and we expect a more graceful decline in performance.

Another classical method  relevant to our inference task  is canonical correlation analysis (CCA)\cite{CCA,Hardoon2004}. CCA can be used to find a pair of  orthogonal projections for mapping data of each modality to  the same space.  The results for this approach are not presented herein for brevity and can be found in \cite{Adali_Thesis}. In terms of  performance for the match detection task, the CCA-based method  was very competitive with our dissimilarity-centric approach.

There have been many efforts toward solving the related problem of \,``manifold alignment". ``Manifold alignment" seeks to find correspondences between disparate datasets in different conditions (which are  sometimes referred to as ``domains'') by aligning their underlying manifolds. A common data setting found in the literature  is the semi-supervised setting\cite{Ham2005a}, where  correspondences between two collections of observations  are given and the task is to find correspondences between a new set of observations in each condition. The proposed solutions\cite{Wang2008,Zhai2010,3wayNMDS}
 follow the common approach of seeking  a common latent space  for multiple conditions  such that the representations  in this space (either projections or embeddings) of the observations match (are commensurate).

Wang and Mahedavan~\cite{Wang2008} suggest an  approach that uses separate embeddings followed by Procrustes Analysis to find maps from the original disparate data to a commensurate space. Given a paired set of points, Procrustes Analysis~\cite{Sibson}  finds a linear transformation from one set of points to the other that minimizes sum of squared distances between pairs. In the problem considered in \cite{Wang2008}, the paired set of points are low-dimensional embeddings of kernel matrices. For the embedding step, they chose to use Laplacian Eigenmaps, though their algorithm allows for any appropriate embedding method.

 Zhai et al.~\cite{Zhai2010}  solves an optimization problem  with respect to  two projection matrices for the observations in two domains.  The energy function that is optimized contains three terms: two \emph{manifold regularization terms} and one \emph{correspondence preserving term}. The  \emph{manifold regularization terms} ensure that the local neighborhood of points are preserved in the low-dimensional space, by making use of the reconstruction error for Locally Linear Embedding\cite{Roweis_LLE}.
The \emph{correspondence preserving term} ensures that ``matched'' points are mapped to close locations in the commensurate space.

Ham et al.~\cite{Ham2005a} solve the problem in the semi-supervised setting by a similar approach, by optimizing a energy function that has three terms that are analogous to the terms in ~\cite{Zhai2010}.

\section{Problem Description}
In the problem setting considered here,  $n$ different objects are measured under $K$ different conditions (corresponding  to, for example, $K$ different sensors). We assume we begin with the data  available in dissimilarity representation. These will be represented in matrix form as $K$ $n \times n$ matrices $\{\Delta_k,k=1 ,\ldots,K\}$.  In addition, for each condition, dissimilarities between  a new object  and the previous
$n$ objects $\{\mathcal{D}_k,k=1 ,\ldots,K\}$ are available. Under  the null hypothesis, these new dissimilarities represent a \emph{single} new object   measured under $K$ different conditions. Under the alternative hypothesis, the dissimilarities $\{\mathcal{D}_k\}$ represent \emph{separate} new objects   measured under $K$ different conditions~\cite{JOFC}. 

For the  Wikipedia  example presented in the introduction, two dissimilarity matrices are available: the dissimilarities between articles based on their textual content ($\Delta_{1})$ and the dissimilarities between the vertices of  the hyperlink graph  ($\Delta_2)$. Various dissimilarity measures defined between pairs of graph vertices can be used to compute $\Delta_2$. The dissimilarities between the new text document and the  previous $n$  text documents $(\mathcal{D}_1)$ are also available, as well as the dissimilarities between  a new vertex added to the graph  and the previous $n$ vertices $(\mathcal{D}_2)$. The null hypothesis is that the new document and the vertex correspond to each other, while the alternative hypothesis is that they are not.

  In order to derive a data representation where dissimilarities from disparate sources ($\{\mathcal{D}_k\}$)  can be compared, the dissimilarities must be embedded in a commensurate metric space where the metric can be used to distinguish between matched and unmatched observations.

 We do not assume the dissimilarity matrices have extra properties beyond the basic definition such as the metric property.  While data from less disparate conditions should result in better performance in the inference task, our method does not restrict the type of dissimilarity measures that are used in each condition. Due to the  disparateness of conditions and the dissimilarity measures used, the dissimilarity matrices $\{\Delta_k\}$ might have different scales of magnitude. We correct for this by normalizing the scale of each  dissimilarity matrix.  A reasonable choice for doing so is to divide each $\{\Delta_k\}$ by its Frobenius norm $\|\Delta_k\|_{F}$.

To embed multiple dissimilarities  $\{\Delta_k\}$  into a commensurate space, an omnibus dissimilarity matrix  $M \in \mathbb{R}^{nk \times nk}$  is constructed. Consider, for $K=2$,
 \begin{equation}
M=  \left[ \begin{array}{cc}
         \Delta_1 & L\\
        L^T  & \Delta_2
     \end{array}  \right]     \label{omnibus}
\end{equation} where $L$ is a matrix of imputed entries to be described later.

\begin{remark}
For clarity of exposition, we will consider $K=2$; the generalization to $K>2$ is straightforward.
\end{remark}

We define the commensurate space to be  $\mathbb{R}^d$, where the embedding dimension $d$ is pre-specified. The selection of $d$ -- model selection -- is  a task that requires much attention and is  beyond the scope of this article. We should emphasize the fact  that the choice of $d$ has an impact on most of the following work.   For applications involving real data, we used the   automatic dimensionality selection heuristic in \cite{ZhuGhodsi}. This heuristic is derived from the assumption that the ordered eigenvalues $\{\lambda_i\}$ are drawn independently from  two  distributions from the same family $f(\lambda_i,\theta_1), f(\lambda_i,\theta_2)$ and optimizes the log-likelihood $$l_q(\lambda)=\sum_{i=1}^d f(\lambda_i,\theta_1)+ \sum_{i=d+1}^n f(\lambda_i,\theta_2) $$ with respect to $d$. Utilizing this heuristic provides an objective method to  determine the ``elbow'' of the scree plot without requiring visualization and inspection by a data analyst.  In the case domain knowledge is available for the data, it should be used to inform the selection of $d$. Otherwise, we recommend  to future practitioners that they select the value of $d$ such that the  data in \emph{both} modalities can be represented with small distortion in $\mathbb{R}^d$.

 We use multidimensional scaling (MDS) \cite{borg+groenen:1997,SMACOF,CMDS} to embed  the omnibus matrix in this  space, and obtain  a configuration of $2n$ embedded points $\{\hat{x}_{ik};1 \leq i \leq n;k\in\{1,2\} \}$ (which can be represented as $\hat{X}$, a $2n \times d$ matrix, where each row of the configuration matrix is the coordinate vector of an embedded point). The discrepancy between the interpoint distances of $\{\hat{x}_{ik}\}$ and the given dissimilarities in  $M$ is made as small  as possible, as measured by an objective function  $$\sigma_W(\widetilde{X};M)=\sum_{k_1\in\{1,2\};k_2\in\{1,2\}} \sum_{1 \leq i_1 \leq n;1 \leq i_2 \leq n }w_{i_1i_2k_1k_2}(d(\hat{x}_{i_1k_1},\hat{x}_{i_2k_2})-M_{i_1i_2k_1k_2})\footnote{which is based on the raw stress criterion $\sum_{1 \leq i_1 \leq n;1 \leq i_2 \leq n}w_{i_1i_2}(d(\hat{x}_{i_1},\hat{x}_{i_2})-\delta_{i_1i_2})^2$} $$  with a $2n \times 2n$ weight matrix $W$. Each entry of $W$, $w_{i_1i_2k_1k_2}$, correspond to an  entry of $M$,$M_{i_1i_2k_1k_2}$.  In matrix form, $ \hat{X}=\arg \min_{\widetilde{X}} \sigma_{W}(\widetilde{X};M)$. This minimization problem is solved by the SMACOF algorithm\cite{SMACOF}.


\begin{remark}
We will use $x_{ik}$ to denote the (possibly notional)  observation  for the $i^{th}$ object in the $k^{th}$ condition, $\tilde{x}_{ik}$ to denote an argument of the objective function  and  $\hat{x}_{ik}$  to denote the $\arg\min$  of the objective function. The notation for configuration matrices ($X,\widetilde{X},\hat{X}$), whose each row  corresponds to  the embedding coordinates of an object, follows the  same convention.
\end{remark}

  Given the omnibus matrix $M$ and the $2n \times d$ embedding configuration matrix $\hat{X}$ in the commensurate space, the out-of-sample extension~\cite{MaThesis} developed for the raw-stress MDS variant will be used to embed the test dissimilarities $\mathcal{D}_1$ and $\mathcal{D}_2$.  Once the test similarities are embedded as two points ($\hat{y}_{1},\hat{y}_{2}$) in  the commensurate space, it is possible to  compute the test statistic \[
\tau=d\left(\hat{y}_{1},\hat{y}_{2}\right)\label{teststat}
\] for the two ``objects'' represented by  $\mathcal{D}_1$ and $\mathcal{D}_2$.  For large values of $\tau$, the null hypothesis will be rejected.
   If  dissimilarities between matched objects are smaller than dissimilarities between unmatched objects with large probability, and the embeddings preserve this stochastic ordering,  we could reasonably expect the test statistic to yield large  power.
\section{Fidelity and Commensurability\label{sec:FidComm}}

Regardless of the inference task,  to expect reasonable performance from the embedded data in the commensurate space,
it is necessary to pay heed to these two error criteria: 

\begin{itemize}
\item Fidelity describes how well the mapping to commensurate space preserves the original dissimilarities. The \emph{loss of fidelity} can be measured with the  within-condition \emph{ fidelity error}, given by
    \[
\epsilon_{f_{(k)}} = \frac{1}{{{n}\choose{2}}} \sum_{1 \leq i < j \leq n} (d(\widetilde{\bm{x}}_{ik},\widetilde{\bm{x}}_{jk})-\delta_{ijkk})^2
.\]
Here $\delta_{ijkk}$ is the dissimilarity between the $i^{th}$ object and the $j^{th}$ object where both objects are in the $k^{th}$  condition, and $\widetilde{\bm{x}}_{ik}$ is the embedded representation of the $i^{th}$ object  for the $k^{th}$ condition;  $d(\cdot,\cdot)$ is the Euclidean distance function.

\item Commensurability describes how well the mapping to commensurate space preserves matchedness of matched observations. The \emph{loss of commensurability} can be measured by the between-condition {\em commensurability error} which is given by
    \[
\epsilon_{c_{(k_1,k_2)}} = \frac{1}{n} \sum_{1 \leq i \leq n;k_1 <k_2} (d(\widetilde{\bm{x}}_{ik_1},\widetilde{\bm{x}}_{ik_2})- { \delta_{iik_1k_2}})^2
\label{comm-error}
\]
 for conditions $k_1$ and $k_2$; $\delta_{iik_1k_2}$  is the dissimilarity between the $i^{th}$ object under  conditions   $k_1$ and  $k_2$.
Although  the between-condition dissimilarities of the same object, ${ \delta_{iik_1k_2}}$, are not available,  it is reasonable to set these dissimilarities to $0$ for all $i,k_1,k_2$. These dissimilarities correspond to  diagonal  entries of the  submatrix $L$ in  the omnibus matrix  $M$ in equation \eqref{omnibus}. Setting these diagonal entries to $0$ forces matched observations to be embedded close to each other. \label{commens}
\end{itemize}

While  the above expressions for  \emph{fidelity} and  \emph{commensurability} errors  are specific to the joint embedding of disparate dissimilarities, the concepts of fidelity and commensurability are  general enough to be applicable to other dimensionality reduction methods for data from disparate sources.

 In addition to fidelity and commensurability, there is the \emph{separability} criteria:  dissimilarities between unmatched  observations in different conditions  should be preserved (so that unmatched pairs are not embedded close together).

Let us now show how fidelity and commensurability errors  can be made explicit in the objective function. Consider the weighted raw stress criterion ($\sigma_{W}(\cdot)$) which we choose as the objective function for the embedding  of  $M$ with a weight matrix $W$. The entries of $M$ are $\{M_{ijk_1k_2}\}$ for the available dissimilarities ($k_1=k_2$). As the between-condition dissimilarities,
$\{M_{ijk_1k_2}\}$ for $i\neq j$, are  not available in general, the entries corresponding to the unavailable dissimilarities can be imputed\footnote{An effective  imputation method is highly dependent on the nature of the disparate dissimilarities. In the case where the dissimilarities have  the same order of magnitude and a high degree of  matchedness, at the very least, the average of the two  dissimilarities will have the right  order of magnitude. It is also possible to disregard these dissimilarities altogether except those between  matched observations in different conditions.} as $
M_{ijk_1k_2}=\frac{M_{ijk_1k_1}+M_{ijk_2k_2}}{2} \label{impute}$. Then the objective function is
\begin{equation}
\sigma_{W}(\widetilde{X};M)=\sum_{i\leq j,k_1\leq k_2} {w_{ijk_1k_2}(D_{ijk_1k_2}(\widetilde{X})-M_{ijk_1k_2})^2  }\label{raw-stress}.
\end{equation}
 Here, $ijk_1k_2$ subscript of a partitioned matrix refers to the entry in the $i^{th}$ row and $j^{th}$ column of the sub-matrix in $k_1^{th}$ row partition and $k_2^{th}$ column partition, $W$ is the weight matrix, $\widetilde{X}$ is the configuration matrix that is the argument of the stress function, $D(\cdot)$ is the matrix-valued function whose outputs are the Euclidean distances between the rows of its matrix argument.   \emph{Each of the individual terms in the sum \textrm{(\ref{raw-stress})} can be ascribed to fidelity, commensurability or separability}. 

\begin{align}
\sigma_W(\cdot;M)  &=  \sum_{i,j,k_1,k_2} \underbrace{{w_{ij{k_1}{k_2}}(D_{ij{k_1}{k_2}}(\cdot)-M_{ijk_1k_2})^2 }}_{term_{i,j,k_1,k_2}}  & \notag\\
\hspace{3pt} &=  \underbrace{\sum_{i=j,k_1<k_2}  term_{i,j,k_1,k_2}}_{Commensurability}  \hspace{10pt}    +  \underbrace{\sum_{i<j,k_1=k_2}   term_{i,j,k_1,k_2}  } _{Fidelity}
\hspace{3pt}+  \underbrace{\sum_{i< j,k_1<k_2}  term_{i,j,k_1,k_2}  } _{Separability}\label{eq:FidCommSep}\hspace{10pt} &.
\end{align}

\begin{remark}
   Due to the fact that data sources are ``disparate", it is not obvious how  a dissimilarity between an object in one condition and another object in another condition  can be computed or  defined in a sensible way. Although these unavailable dissimilarities appearing in the separability term can be imputed as mentioned,  MDS variants such as weighted raw stress allows for disregarding the between-condition dissimilarities, by setting the corresponding weights  in the raw stress function  to 0.   By doing so, we restrict our attention to the fidelity-commensurability tradeoff.
\end{remark}

 As mentioned in description of commensurability, we set the between-condition dissimilarities of the same object ($\{M_{iik_1k_2}\}$) to $0$. Then the raw stress function can be written as
\begin{align}
\sigma_W(\widetilde{X};M)\hspace{3pt}
\hspace{3pt}&=&\underbrace{\sum_{i=j,k_1< k_2}  {w_{ij{k_1}{k_2}}(D_{ij{k_1}{k_2}}(\widetilde{X}))^2}}_{Commensurability}  \hspace{10pt}  &  +&\underbrace{\sum_{i< j,k_1=k_2}  {w_{ij{k_1}{k_2}}(D_{ij{k_1}{k_2}}(\widetilde{X})-M_{ijk_1k_2})^2  }  } _{Fidelity}\notag\label{eq:FidCommSep_final}\hspace{10pt}.
\end{align}
This motivates  the naming of the   omnibus embedding approach as Joint Optimization of Fidelity and Commensurability (JOFC).

 The weights in the raw stress function allow us to address the question of the optimal tradeoff of  fidelity and commensurability. Let $w \in (0,1)$. Setting the weights ($\{w_{ijk_1k_2}\}$)  for the commensurability  and fidelity  terms    to $w$ and $1-w$, respectively,  will allow us to control the relative importance of fidelity and commensurability terms in the objective function.

 Let us denote the raw stress function with these simple weights by $\sigma_w(\widetilde{X};M)$. With simple weighting, when $w=0.5$, all terms in the objective function have the same weights. We will refer to this weighting scheme as \emph{uniform weighting}. Uniform weighting does not necessarily yield the best fidelity-commensurability tradeoff in terms of subsequent inference.

 Previous investigations of the JOFC approach \cite{JOFC,MaThesis} did not consider the effect of non-uniform weighting.
Our thesis is that using non-uniform weighting  in the objective function will allow for superior performance.
That is, for a given inference task there is an optimal $w$ for inference, denoted by $w^*$, and in general $w^* \neq 0.5$.
In particular, as our inference task, we consider hypothesis testing, as in \cite{JOFC},
and we let the area under the ROC curve, $AUC(w)$, be our measure of performance for any $w \in [0,1]$.
In this setting, we show that $AUC(w)$ is continuous in the interval $(0,1)$, and hence $w^* = \arg\max_{w \in (\epsilon,1-\epsilon)} AUC(w)$ exists for arbitrarily small $\epsilon$.
We demonstrate the potential practical advantage of our weighted generalization of JOFC via simulations.

\section{Definition of  $w^{*}$}

\begin{remark}
In our notation, $(.)$  denotes  either $(m)$  or   $(u)$. In the former case, an expression refers to values under  ``matched'' hypothesis, in the latter, the expression refers to values under  ``unmatched''   hypothesis.
\end{remark}
Let us denote the test dissimilarities ($\mathcal{D}_1$, $\mathcal{D}_2$)  by  ($\mathcal{D}_1^{(m)}$, $\mathcal{D}_2^{(m)}$)  under the  ``matched'' hypothesis, and  by ($\mathcal{D}_1^{(u)}$, $\mathcal{D}_2^{(u)}$)  under the alternative. The out-of-sample embedding of ($\mathcal{D}_1^{(m)}$, $\mathcal{D}_2^{(m)}$) involves the  augmentation of  the omnibus matrix $M$, which consists of $n$ matched  pairs of dissimilarities,  with ($\mathcal{D}_1^{(m)}$, $\mathcal{D}_2^{(m)}$). The resulting augmented  $(2n+2)\times (2n+2)$ matrix  has the form:

 \begin{equation}
\Delta^{(m)}=  \left[ \begin{array}{cccc}
          \multicolumn{2}{c}{\multirow{2}{*}{\Huge{$M$}}} &  \mathcal{D}_1^{(m)} &\vec{\mathcal{D}}_{NA}\\
        & &  \vec{\mathcal{D}}_{NA}   & \mathcal{D}_2^{(m)} \\
				\mathcal{D}_1^{(m)T} & \vec{\mathcal{D}}_{NA}^T  &  0 & \mathcal{D}_{NA} \\
         \vec{\mathcal{D}}_{NA}^T & \mathcal{D}_2^{(m)T} & \mathcal{D}_{NA} &0\\
     \end{array}  \right].     \label{aug_omnibus}
\end{equation}  where
the scalar $\mathcal{D}_{NA}$ and    $\vec{\mathcal{D}}_{NA}$ (a vector of NAs that has length $n$)   represent dissimilarities that are not available.
In our JOFC procedure, these unavailable entries in $\Delta^{(m)}$ are either imputed using other dissimilarities that are available, or ignored in the embedding optimization. For a simpler  notation, let us assume it is the former case. Also note that $\Delta^{(u)}$  has the same form as $\Delta^{(m)}$ where $\mathcal{D}_k^{(m)}$ is replaced by $\mathcal{D}_k^{(u)}$.

We define the dissimilarity matrices \{$\Delta^{(m)},\Delta^{(u)}$\} to be  two matrix-valued random variables: $\Delta^{(m)}:\Omega \rightarrow \mathbf{M}_{(2n+2)\times (2n+2)} $ and  $\Delta^{(u)}:\Omega \rightarrow \mathbf{M}_{(2n+2)\times (2n+2)} $) for the appropriate sample  space $(\Omega)$.
\begin{remark}
Suppose the objects in $k^{th}$  condition  can be represented as points in a measurable space $\Xi_k$, and the dissimilarities in $k^{th}$ condition are given by  a dissimilarity measure $\delta_k$ acting on pairs of points in $\Xi_k$. Assume $\mathcal{P}_{(m)}$ is the joint probability distribution over matched objects, while the joint distribution of unmatched objects \{$k=1,\ldots,K$\}  is $\mathcal{P}_{(u)}$. Assuming the data are i.i.d., under the two hypotheses (``matched'' and ``unmatched'', respectively), the $n+1$ pairs of objects are governed  by the product distributions $\{\mathcal{P}_{(m)}\}^n \times \mathcal{P}_{(m)} $ and $\{\mathcal{P}_{(m)}\}^n \times \mathcal{P}_{(u)} $.  The distributions of $\Delta^{(m)}$ and $\Delta^{(u)}$ are the induced probability distributions of  these product distributions (induced by the  dissimilarity measure $\delta_k$ applied to  objects in $k^{th}$ condition \{$k=1,\ldots,K$\}).
\end{remark}


 We now consider the embedding of $\Delta^{(m)}$ and $\Delta^{(u)}$ with the criterion function  $\sigma_W(\widetilde{X}; \Delta^{(.)})$. The arguments of the function are  $\widetilde{X}= \left[
\begin{array}{c}
{\widetilde{\mathcal{T}}} \\
\widetilde{y}_{1}^{(.)} \\
\widetilde{y}_{2}^{(.)}
\end {array}
\right]$ where ${\widetilde{\mathcal{T}}}$ is the argument for the in-sample embedding of the first $n$ pairs of matched points, and
 $\{\widetilde{y}_{1}^{(.)} \}$ and $\{\widetilde{y}_{2}^{(.)} \}$ are the arguments for the embedding coordinates of the matched  or unmatched pair,
and the omnibus dissimilarity matrix $\Delta^{(.)}$ is equal to  $\Delta^{(m)}$  (or $\Delta^{(u)}$) for the embedding of the  matched (unmatched) pair. Note that we use the simple weighting scheme, so with a slight abuse of notation, we rewrite the criterion function as  $\sigma_w(\widetilde{X}; \Delta^{(.)})$ where $w$ is a scalar parameter.
The embedding coordinates for the matched or unmatched pair  ${\hat{y}_{1}^{(.)},\hat{y}_{2}^{(.)}}$ are given by
 \[
{\hat{y}_{1}^{(.)},\hat{y}_{2}^{(.)}}
=\argmin_{\widetilde{y}_{1}^{(.)}, \widetilde{y}_{2}^{(.)}}\left[\min_{\widetilde{\mathcal{T}}}
{\sigma_w\left(
\left[
\begin{array}{c}
{\widetilde{\mathcal{T}}} \\
\widetilde{y}_{1}^{(.)} \\
\widetilde{y}_{2}^{(.)}
\end {array}
\right]
,
\Delta^{(.)}
\right)
}
\right].
\]

\begin{remark}
 Note that the in-sample embedding of $\widetilde{\mathcal{T}}$ is necessary but irrelevant for the inference task, hence the minimization with respect to $\widetilde{\mathcal{T}}$ is denoted by  $\min$ instead $\argmin$. It can be considered as a nuisance parameter for our hypothesis testing.
\end{remark}
\begin{remark}
 Note also that  all of the random variables following the embedding, such as $\{\hat{y}_{k}^{(.)}\}\!$,  are dependent on $w$; for the sake of simplicity, this will  be suppressed in the notation.
\end{remark}

 Under reasonable assumptions, the embeddings $\Delta^{(m)} \rightarrow  \{\hat{y}_{1}^{(m)},\hat{y}_{2}^{(m)}\!\}$  and $\Delta^{(u)}\rightarrow \{\hat{y}_{1}^{(u)} , \hat{y}_{2}^{(u)}\}$ are measurable maps for all $w \in (0,1)$ \cite{measurable_Niemiro1992}. Then, the distances between the embedded points are random variables and we can define the test statistic $\tau$ as $d(\hat{y}_{1}^{(m)},\hat{y}_{2}^{(m)})$ under the null hypothesis  and $d(\hat{y}_{1}^{(u)},\hat{y}_{2}^{(u)})$ under the alternative. Under the null hypothesis, the distribution of the statistic is governed by the distribution of $\hat{y}_{1}^{(m)}$ and $\hat{y}_{2}^{(m)}$; under the alternative it is governed by  the distribution of $\hat{y}_{1}^{(u)}$ and $\hat{y}_{2}^{(u)}$.

 Then, the statistical power as a function of $w$ is given by  \[\beta\left( w,\alpha\right)=1-F_{d \left(\hat{y}_{1}^{(u)},\hat{y}_{2}^{(u)}\right)} \left(F_{d\left(\hat{y}_{1}^{(m)},\hat{y}_{2}^{(m)}\right)}^{-1}(1-\alpha) \right)\] where $F_Y$ denotes  the   cumulative distribution function of  $Y$. The area under curve (AUC) measure  as a function of $w$ is defined as
\begin{equation}
AUC(w)=\int_{0}^{1}\! \beta\left( w,\alpha\right)\,\mathrm{d}\alpha \; . \label{AUC_def}
\end{equation}
Although we might care about optimal $w$ with respect to  $\beta\left( w,\alpha\right)$ (with a fixed Type I error rate $\alpha$),  it will be more convenient to define $w^*$ in terms of the AUC function.

 Finally, define $$w^{*}=\arg\max_w{AUC\left( w\right)}. $$

 Some important questions about $w^*$ are  related to the nature of the AUC function.
While finding an analytical expression for the value of $w^*$ is intractable, an estimate $\hat{w}^*$  based on  estimates of $AUC(w)$ 
 can be computed.  For the Gaussian setting described in  section \ref{subsec:GaussianSet}, a Monte Carlo simulation is performed  to find the estimate of $AUC(w)$ for different values of $w$.

\subsection{Continuity of $AUC(\cdot)$}
 Let $T_0(w)=d(\hat{y}_{1}^{{(m)}},\hat{y}_{2}^{{(m)}})$ and $T_a(w)=d(\hat{y}_{1}^{(u)},\hat{y}_{2}^{(u)})$ denote the value of the test statistic under null and alternative distributions  for the embedding with the simple weighting $w$.  
The AUC function can be written as $$AUC(w)=P\left[T_a(w)>T_0(w)\right]$$ where $T_a(\cdot)$ and $T_0(\cdot)$  can be regarded as  stochastic processes whose sample paths are functions  of $w$.  We will prove that $AUC(w)$ is continuous with respect to $w$.
We start with this lemma from \cite{Raik1972}.

\begin{lem}
Let $z$ be a random variable. The functional $g(z;\gamma) = P\left[z \geq \gamma \right]$ is upper semi-continuous in probability with respect to $z$. Furthermore, if $P\left[z = \gamma \right]=0$, $g(z;\gamma)$ is continuous in probability with respect to $z$. \label{lemma_1}
\end{lem}

\begin{proof}
Suppose $z_n$ converges to $z$ in probability. Then by definition, for any  $\delta>0$ and  $\epsilon>0$, $\exists	N\in\mathbb{Z^+}$ such that for all   $n \geq N$
$$ Pr\left[\left|z_n-z\right| \geq \delta \right] \leq \epsilon.$$

 The functional  $g(z;\gamma)$ is  non-increasing with respect to $\gamma$. Therefore, for $\delta>0$,
$g(z_n;\gamma) -g(z;\gamma) \geq g(z_n;\gamma) -g(z;\gamma-\delta) $. Furthermore, $g(z;\gamma)$ is left-continuous with respect to $\gamma$, so the difference between the two sides of the inequality can be made as small as desired.

\begin{eqnarray}
g(z_n;\gamma) - g(z;\gamma-\delta) & = &Pr\left[z_n\geq \gamma \right] -Pr\left[z \geq  \gamma - \delta \right] \label{prob_defn}\\
& \leq &  Pr\left[\{z_n \geq \gamma \} \backslash \{z \geq \gamma - \delta \} \right] \label{set_diff}\\
& \leq & Pr\left[\{\{z_n \geq \gamma \} \backslash \{z \geq \gamma - \delta \} \} \cap \{z_n \geq  z\} \right] \label{conjunct_with_set} \\
& =  & Pr\left[\{z_n - z \geq \delta \} \right] \leq \epsilon. \label{upper_semicont}
\end{eqnarray}

Since $\epsilon$ and $\delta$ are arbitrary,
 $ \limsup_{n \rightarrow \infty} ( {g(z_n;\gamma)}- g(z;\gamma) ) =  0$ for any $\delta>0, $ i.e. $g(z;\gamma)$ is upper semi-continuous.

By arguments  symmetric to  \eqref{prob_defn}-\eqref{upper_semicont}, we can show that

\begin{equation}
g_(z;\gamma+\delta) - g(z_n;\gamma) \leq \epsilon. \label{lower_semicont}
\end{equation}

In addition, assume that  $P\left[z = \gamma \right]=0$. Then, $g(z;\gamma)$ is also right-continuous with respect to $\gamma$. Therefore,
$g(z_n;\gamma) -g(z;\gamma) \leq g(z_n;\gamma) -g(z;\gamma+\delta)$ and the difference between the two sides of the inequality can be made as small as possible.
Along with \eqref{lower_semicont}, this means that

 \[
\liminf_{n \rightarrow \infty} ( {g(z_n;\gamma)}- g(z;\gamma) ) = 0.
\] Therefore, $\lim_{n\rightarrow \infty}g(z_n;\gamma) = g(z;\gamma)$, i.e. $g(z;\gamma)$ is continuous in probability with respect to $z$.
\end{proof}

\begin{thm} \label{main_thm}
Let $T(w)$ be  a stochastic process indexed by $w$ in the interval (0,1). Assume  the process is continuous in probability  (stochastic continuity)   at $w=w_0$,  i.e.
\begin{equation} \forall a>0 \quad  \lim_{s \rightarrow w_0} Pr\left[\left|T(s)-T(w_0) \right| \geq a \right] = 0
\end{equation}
for $ w_0\in (0,1)$. Furthermore, assume that $Pr\left[T(w_0)=0\right]=0$.

Then, $Pr \left[ T(w) \geq 0\right]$ is continuous at $w_0$.
\end{thm}

\begin{proof}
Consider any sequence $w_n \rightarrow w_0$. Let $z_n = T(w_n)$ and  $z=T(w_0)$ and choose $\gamma=0$. Since $T(w)$ is continuous in probability at $w_0$ and $Pr\left[T(w_0)=0\right]=0$, conditions for Lemma \ref{lemma_1} hold, i.e.\hspace{-.5em} as $w_n\rightarrow w_0$, $z_n$ converges in probability to $z=T(w_0)$. By  Lemma \ref{lemma_1}, we conclude  $g(T(w_n); 0 ) = Pr \left[ T(w_n) \geq 0\right]$ converges to $g(T(w_0);0)$. Therefore $g(T(w);0)$ is continuous with respect to $w$.
\end{proof}

\begin{cor}{
 If $Pr[T_a(w)-T_0(w)=0]=0$ and $T_a(w)$, $T_0(w)$ are continuous in probability for all $w \in (0,1)$, then $AUC(w)=Pr\left[T_a(w)-T_0(w) >0 \right]$ is continuous with respect to $w$  in the interval $(0,1)$.}
\end{cor}
\begin{proof}

Let $T(w)=T_a(w)-T_0(w).$ Then Theorem \ref{main_thm} applies everywhere in the interval (0,1).
\end{proof}

In any closed interval that is a subset of $(0,1)$, the AUC function is continuous and therefore attains its global maximum in that closed interval.

 We do not have closed-form expressions for the null and alternative distributions of the test statistic $\tau$ (as a function of  $w$), so we cannot provide a rigorous proof of the uniqueness of $w^*$. However, for various data settings, simulations always resulted in \emph{unimodal}  estimates for the AUC function which indicates a unique $w^*$ value.

\section{Simulation Results\label{sec:Simulation Results}}
\subsection{Gaussian setting\label{subsec:GaussianSet}}

  Let $n$ ``objects" be represented  by  $\bm{\alpha}_i \sim^{iid} \mathcal{N}(\bm{0},I_p)$.  Let the $K=2$ measurements for the $i^{th}$ object under the different conditions ($k\in(1,2)$) be denoted  by $\bm{x}_{ik}  \sim^{iid} \mathcal{N}(\bm{\alpha_i},\Sigma)$. The covariance matrix
  $\Sigma$ is a positive-definite $p\times p$ matrix whose maximum eigenvalue is   $\frac{1}{r} $. See Figure~\ref{fig:Fig1}.

 Dissimilarities ($\Delta_1$ and  $\Delta_2$) for the omnibus embedding are the Euclidean distances between the measurements in the same condition.

The parameter $r$ controls the variability between ``matched" measurements. If $r$ is large, it is expected that the distance between matched measurements
$\bm{x}_{i1}$ and $\bm{x}_{i2}$ is stochastically smaller than $\bm{x}_{i1}$ and $\bm{x}_{i'2}$ for $i \neq i'$ ; if $r$ is small, then dissimilarities  between pairs of ``matched"  measurements and
``unmatched'' are less distinguishable. Therefore, a smaller value of $r$ makes the decision problem harder, as  the test statistic under null and alternative will have highly similar distributions, resulting in  higher rate of errors or tests with smaller AUC measure.

    \begin{figure}
	\begin{center}
    \includegraphics[scale=1]{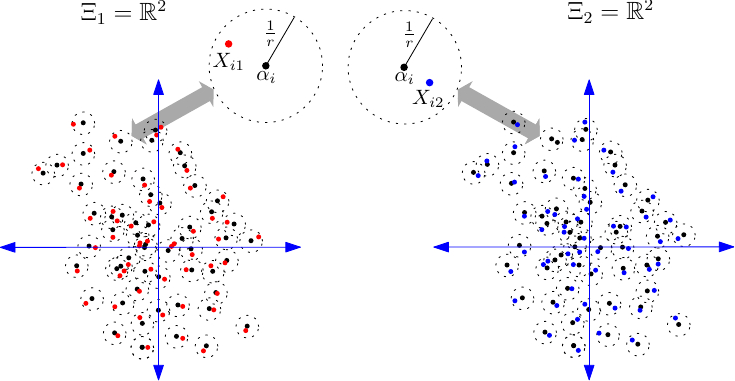}
    \caption{For the  Gaussian setting (Section \ref{subsec:GaussianSet}), the objects can be represented by $\bm{\alpha_i}$  which are two-dimensional random vectors denoted by black points and distributed as $\mathcal{N}(\bm{0},I_p)$. The dashed lines show the equal probability contours for each  $\bm{\alpha_i}$. Since  the measurements in the two conditions and the original object   are in the same space ($\mathbb{R}^2$), $\bm{\alpha_i}$ can be shown along with the measurements  $\bm{x}_{ik}$ which are denoted by red ($k=1$) and blue ($k=2$) points respectively.}
\label{fig:Fig1}
	\end{center}
  \end{figure}

\subsection{Simulation\label{subsec:sim}}

We generate the training data of matched sets of measurements according to  the Gaussian setting. Dissimilarity representations are computed from pairwise Euclidean distances of these measurements. We also generate a set of matched pairs and unmatched pairs of measurements for testing using the same Gaussian setting. Following the out-of-sample embedding of the test dissimilarities
we compute test statistics  for matched and unmatched pairs. This allows us to compute the empirical power  at different values of $\alpha$ (Type I error rate)  and the empirical AUC measure.

 The measurements for the Gaussian setting are vectors in $p$-dimensional Euclidean space ($p$=5). For $nmc=400$ Monte Carlo replicates,  $n=150$ matched training pairs and $m=250$ matched and unmatched test pairs (generated according to the Gaussian setting) were generated. Using the resulting test statistic values for matched and unmatched test pairs, the AUC measure was computed for different $w$ values along with the average of the power ($\beta$) values at  different values of $\alpha$.

 The plot in Figure \ref{fig:MVN-c0-power-alpha} shows the  $\beta$ vs $\alpha$ curves for different values of  $w$. It is clear from the plot that $w$ has a significant effect on statistical power ($\beta$). There are several $w$ values in the range $(0.85,0.95)$ that result in power values that are close to optimal, and statistical power declines as $w \rightarrow 0$ or as $w \rightarrow 1$.
Also, note that   the estimate of the optimal $w^{*}$  has  an AUC measure higher than  that of $w$=0.5 (uniform weighting) which confirms our thesis that non-uniform weighting could result in larger statistical power. This finding was confirmed using data generated
according to  the Gaussian setting with different sets of parameters.

\begin{figure}[h!tb]
     \centering
\includegraphics[scale=0.58,angle=0]{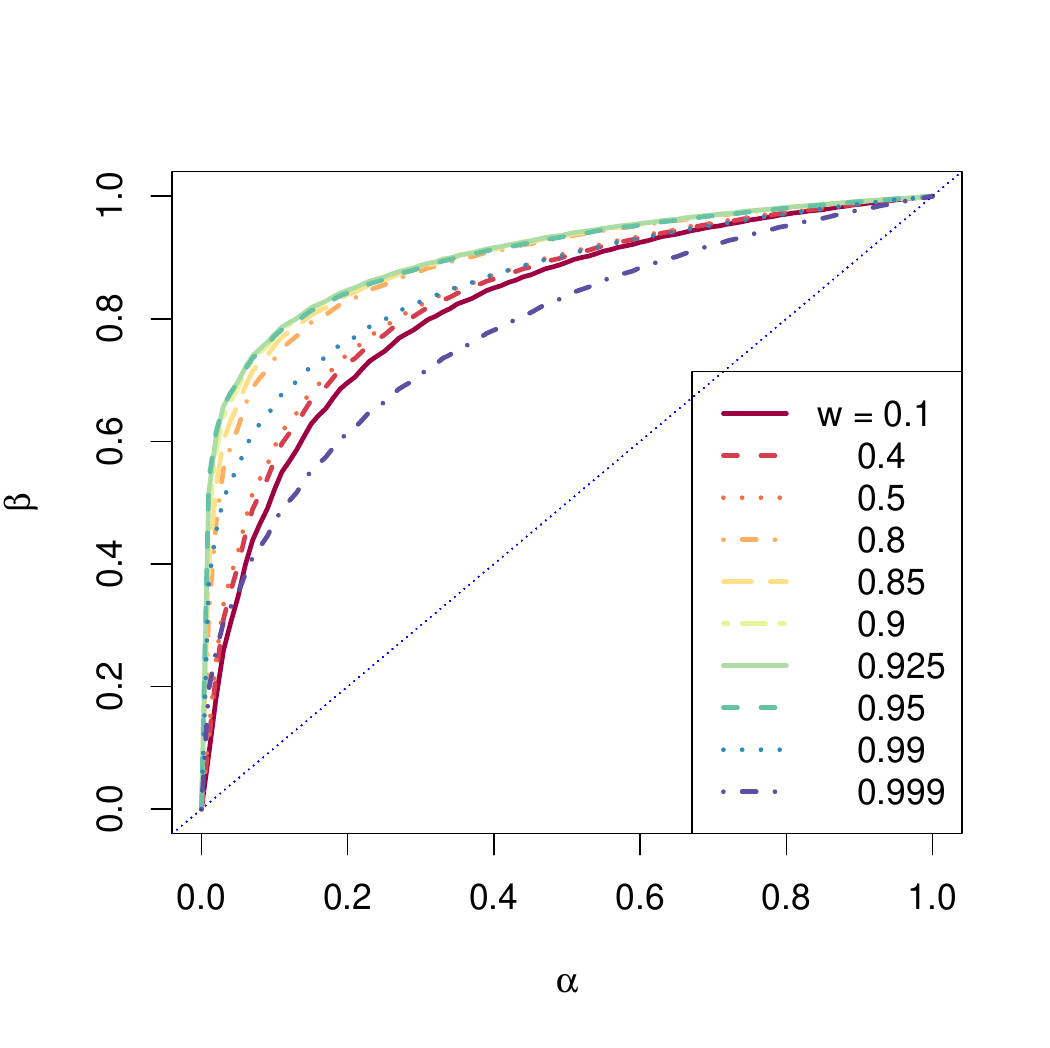}
\caption{$\beta$ vs $\alpha$  for different choices of $w$}
\label{fig:MVN-c0-power-alpha}
\end{figure}

 In Figure
 \ref{fig:MVN-c0-power-w},  $\beta(w)$ is plotted against $w$ for fixed values of $\alpha$.  Here the effect of $w$ on power can be seen more clearly: for all  three values of $\alpha$, $\beta(w)$ increases as $w$ approaches a value in the range (0.91,0.96) and then starts to decrease. We see this trend for different values of $\alpha$ which is consistent with our conjecture that the AUC function, which is defined in equation \eqref{AUC_def}, is unimodal.

\begin{figure}[htb]
      \centering
         \includegraphics[scale=.54]{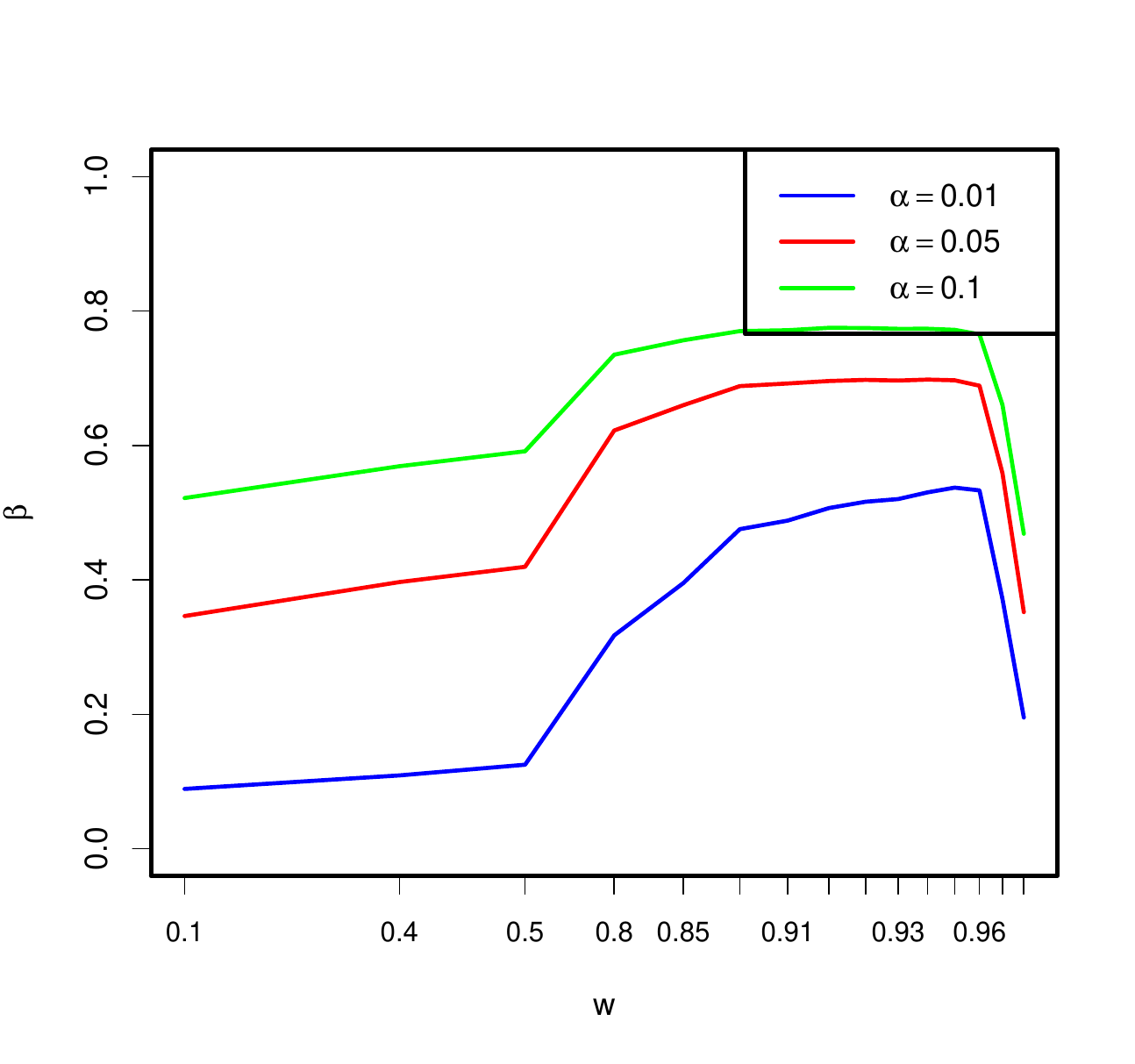}
\caption{$\beta$ vs $w$ plot for different choices of  $\alpha$}
\label{fig:MVN-c0-power-w}

  \end{figure}
The average AUC measure for these $nmc=400$ Monte Carlo replicates are  in  Table \ref{tab:AUCW}.

\begin{figure}[htb]
	\centering

		\includegraphics[scale=.63]{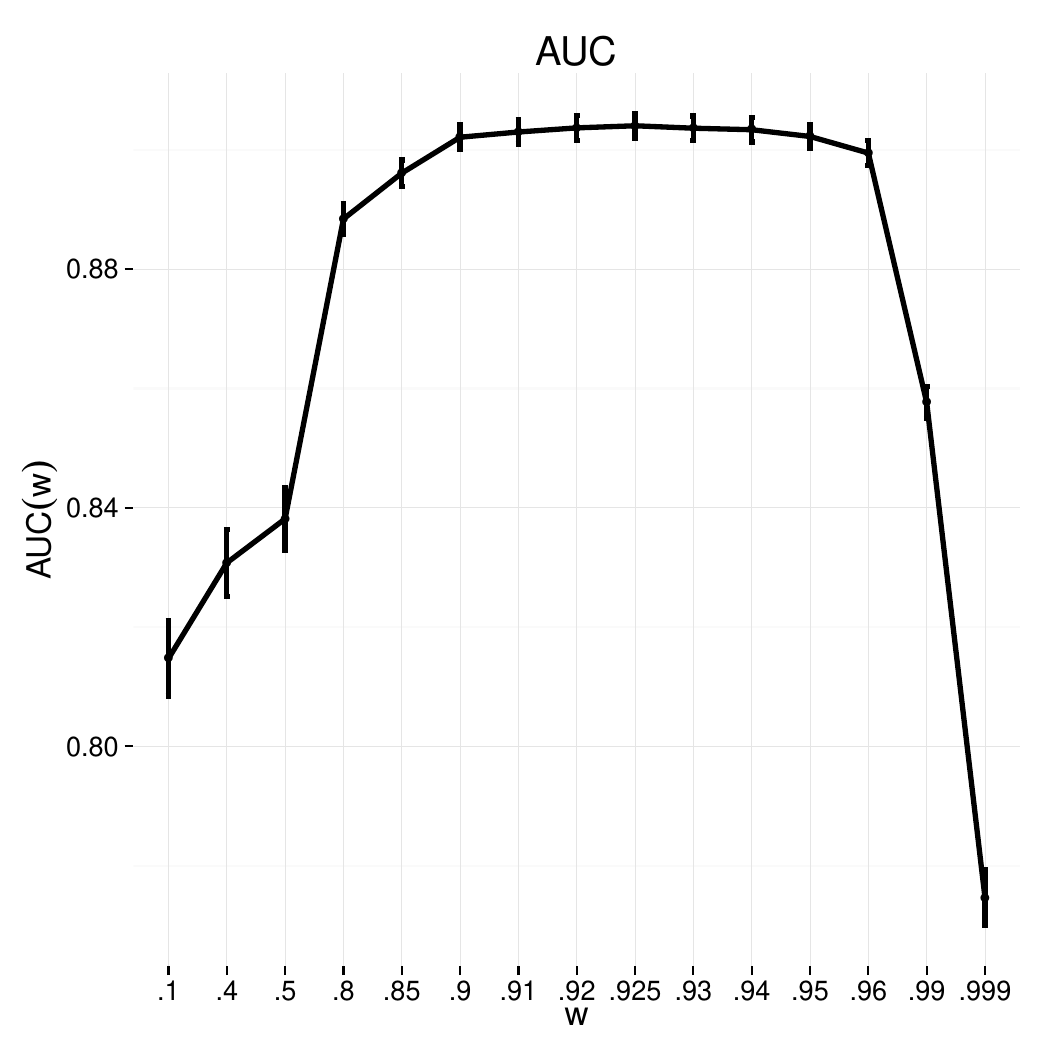}

	\caption{Mean and SE of $AUC(w)$ values for 400 Monte Carlo replicates.}
	\label{fig:AUCW}
\end{figure}

\begin{table}[tb]
\centering
\begin{tabular}{rrrrrrrrr}
  \hline
 $w$ & 0.1 & 0.4 & 0.5 & 0.8 & 0.85 & 0.9 & 0.91 & 0.92 \\
  \hline
mean & 0.8147 & 0.8308 & 0.8381 & 0.8884 & 0.8961 & 0.9021 & 0.9030 & 0.9037  \\
  SE & 0.0640 & 0.0574 & 0.0537 & 0.0258 & 0.0226 & 0.0209 & 0.0206 & 0.0210   \\
	\hline
	  $w$ & 0.925& 0.93 & 0.94 & 0.95 & 0.96 & 0.99 & 0.999 & \\
		\hline
		mean & 0.9040 & 0.9036 & 0.9034 & 0.9022 & 0.8995 & 0.8576 & 0.7746 & \\
		SE & 0.0209 & 0.0209 & 0.0210 & 0.0210 & 0.0217 & 0.0270 & 0.0474 &\\
   \hline
\end{tabular}
\caption{Mean and standard error of  $AUC(w)$ for 400 Monte Carlo replicates.}
\label{tab:AUCW}
\end{table}

  The value of $w$ which results in the highest AUC measure  is different for each Monte Carlo replicate.  The number of  replicates  for  which a particular $w$ value led to the highest AUC measure is shown in  the bar chart in  Figure \ref{fig:ArgMaxWAUCW}. Only the non-zero counts are  shown in the plot. The estimate $\hat{w}^*$ can be chosen as   0.925, as it is the mode of $w^*$ estimates from all replicates. We should note that the AUC function is very flat in the interval $(0.85,0.99)$, and it is possible that the difference between the largest value of the AUC measure and the next largest is very small for any replicate.
\begin{figure}[h!]
	\centering

		\includegraphics[scale=.53]{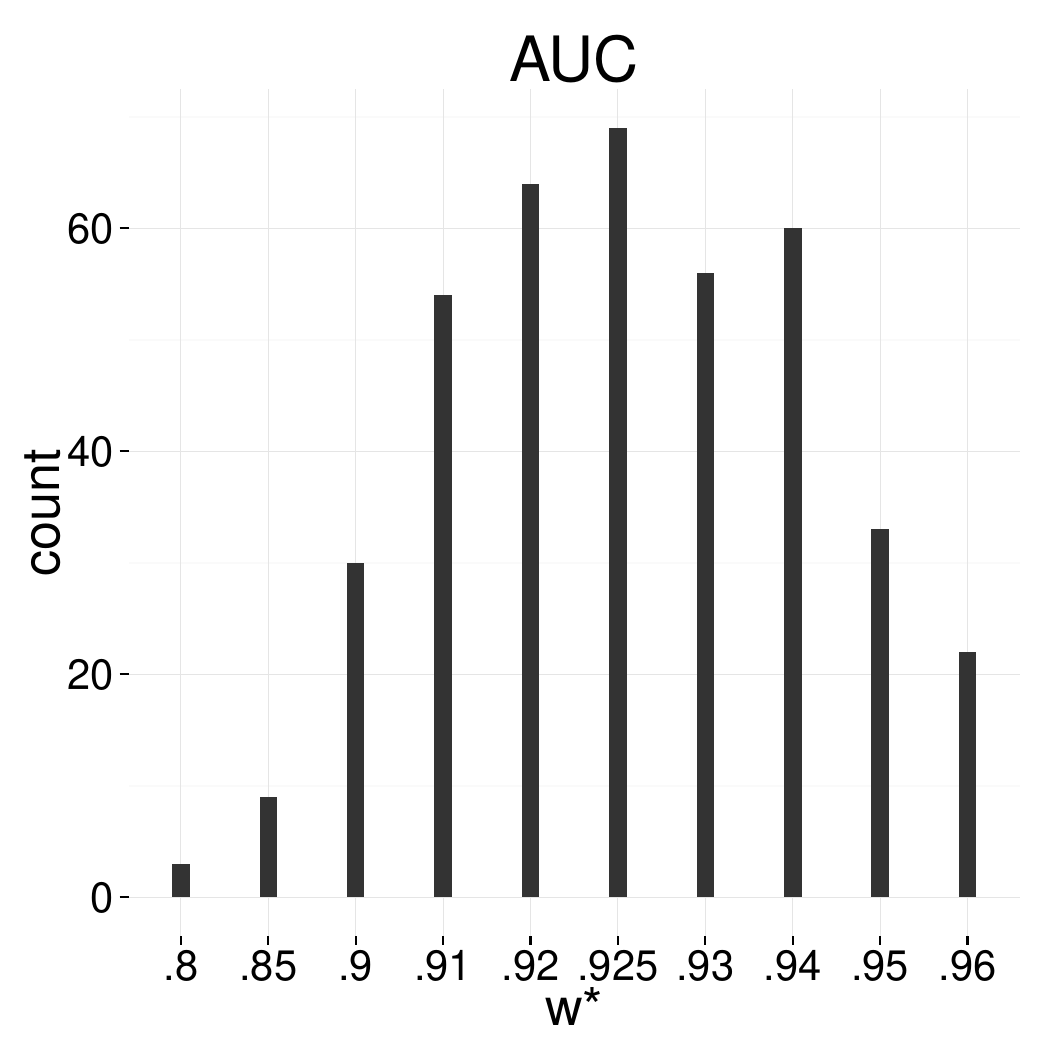}

	\caption{Frequency plot of $w^*$ estimates for 400 replicates.}
	\label{fig:ArgMaxWAUCW}
\end{figure}


\begin{figure}[htb]
     \centering
\includegraphics[scale=0.62,angle=0]{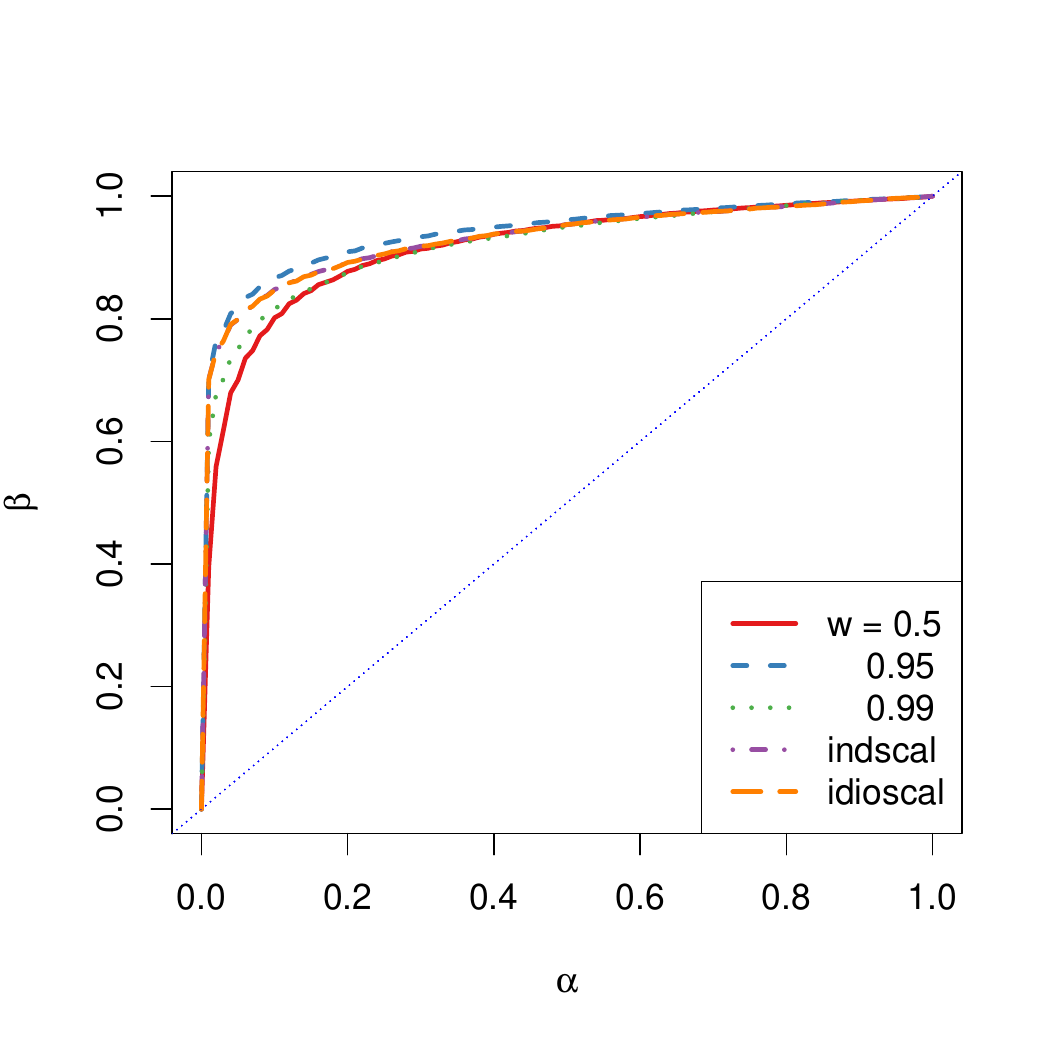}
\caption{$\beta$ vs $\alpha$  for JOFC with different choices of $w$ and for \mbox{INDSCAL} method}
\label{fig:indscal_compare}
\end{figure}

We also compared our approach with two 3-way MDS-based methods, \mbox{IDIOSCAL} and \mbox{INDSCAL}, for the Gaussian setting. Our JOFC approach with the best $w$ value outperforms both as shown in Figure \ref{fig:indscal_compare}. The ROC curves for \mbox{INDSCAL} and  \mbox{IDIOSCAL} are indistinguishable  as the  estimated transformation matrix from group space to configuration spaces should not deviate from the identity  matrix too much\footnote{For the configuration space of dimension $p$, \mbox{INDSCAL} restricts the transformation matrix to  $p \times p$ diagonal matrices   while \mbox{IDIOSCAL} allows for any $p \times p$  square matrix.}. Thus,  both methods are equivalent in this Gaussian setting with identity mapping between the two conditions.
\section{Experiments with multimodal data}
To test our approach, we will use a real dataset with two disparate conditions. The dissimilarities in the two conditions\footnote{The dissimilarity data are available in \url{http://www.cis.jhu.edu/~parky/CGP/cgp.html}.}  are derived from brain fMRI images collected from $n=42$ patients, and from the  personality scores of those patients \cite{RealData}. We wish to discover commensurate mappings for data from the two domains and match detection provides a simple test problem for measuring the commensuracy of the joint embedding.

 For each replicate of our experiment, we randomly sample two matched  pairs of rows from the two dissimilarity matrices. We use one of the matched pairs as our matched test example. We also  use one row  of each pair (say, first row from the first condition and second row from the second condition) as the unmatched test example. The remaining rows/columns form the in-sample dissimilarity matrices ${\Delta}_1$, ${\Delta}_2$. We jointly embed these 40 $\times$ 40 dissimilarity matrices. The test examples are then OOS-embedded to compute the test statistic for match detection. The critical value for the decision to reject the null hypothesis can be computed by a boot-strapping method, i.e. repeatedly sampling two matched row pairs of ${\Delta}_1$ and ${\Delta}_2$ to OOS-embed, so that the test statistic under the null and alternative hypothesis can be computed\footnote{The common bootstrapping method of resampling observations is inappropriate in this case, since the dissimilarity matrices would have rows of zeroes.}.

The ROC curve in figure \ref{fig:real_MRIcogndata_ROC} shows the performance of the JOFC approach. We also show the results for the PrM method which uses Procrustes matching of separate MDS embeddings for a commensurate representation.  For low $\alpha$ values, which are typically of more interest, the JOFC approach out-performs  PrM. We chose  $\alpha =0.05$ and used critical values computed  via bootstrapping in each replicate of the experiment. The effective size averaged over the replicates is 0.04, the average effective power is 0.18.

\begin{figure}[h]
	\centering
		\includegraphics[scale=.85]{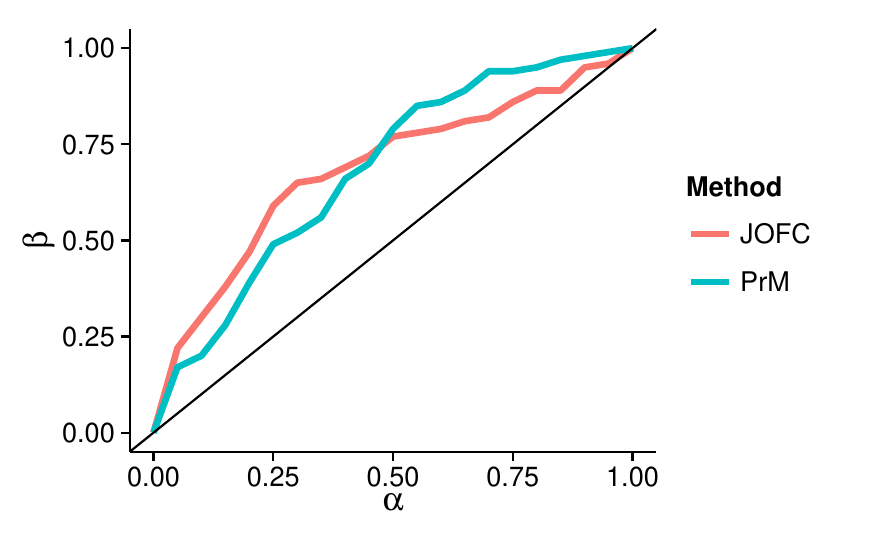}
	\caption{ROC curve for multimodal dataset: MRI/cognitive test }
	\label{fig:real_MRIcogndata_ROC}
\end{figure}

\section{Conclusion}
 We investigated the tradeoff between fidelity and commensurability and its relation to the weighted raw stress criterion for MDS.
   For  hypothesis testing as the exploitation task, different values of the tradeoff parameter $w$ were compared in terms of testing power.
    The results indicate that when doing a joint optimization, one should consider an optimal compromise point between fidelity and commensurability,
       which corresponds to an optimal weight $w^*$ of the weighted raw stress criterion in contrast to the uniform weighting. We consider an estimate of $w^*$ chosen from a finite set of $w$ values for a data generated according to the Gaussian setting. We test the applicability of the JOFC approach for a real multimodal dataset and find it provides satisfactory performance.

\bibliographystyle{plain}
\bibliography{FidComm}

\begin{thebibliography}{10}

\bibitem{DISTATIS}
Herv{\'e} Abdi, Dominique Valentin, Alice~J O'Toole, and Betty Edelman.
\newblock {DISTATIS}: The analysis of multiple distance matrices.
\newblock In {\em Proceedings of the IEEE Computer Society: International
  Conference on Computer Vision and Pattern Recognition}, pages 42--47, 2005.

\bibitem{Adali_Thesis}
Sancar Adali.
\newblock {\em Joint Optimization of Fidelity and Commensurability for Manifold
  Alignment and Graph Matching}.
\newblock PhD thesis, The Johns Hopkins University, 2014.

\bibitem{RealData}
Jonathan~S. Adelstein, Zarrar Shehzad, Maarten Mennes, Colin~G. DeYoung,
  Xi-Nian Zuo, Clare Kelly, Daniel~S. Margulies, Aaron Bloomfield, Jeremy~R.
  Gray, Xavier Castellanos, and Michael Milham.
\newblock Personality is reflected in the brain's intrinsic functional
  architecture.
\newblock {\em PLoS ONE}, 6(11):e27633, 11 2011.

\bibitem{borg+groenen:1997}
Ingwer Borg and Patrick Groenen.
\newblock {\em Modern Multidimensional Scaling: Theory and Applications}.
\newblock Springer, 1997.

\bibitem{Carroll_Chang_1970}
J.~Douglas Carroll and Jih-Jie Chang.
\newblock Analysis of individual differences in multidimensional scaling via an
  n-way generalization of ``{Eckart}-{Young}'' decomposition.
\newblock {\em Psychometrika}, 35(3):283--319, 1970.

\bibitem{3wayNMDS}
Brent Castle, Michael~W. Trosset, and Carey~E. Priebe.
\newblock A nonmetric embedding approach to testing for matched pairs.
\newblock (TR-11-04), October 2011.

\bibitem{proxscal}
Jacques Commandeur and Willem~J. Heiser.
\newblock {Mathematical Derivations in the Proximity Scaling ({PROXSCAL}) of
  Symmetric Data Matrices}.
\newblock Technical Report Research Report RR-93-04, Department of Data Theory,
  Leiden University, 1993.

\bibitem{SMACOF}
Jan de~Leeuw.
\newblock Applications of convex analysis to multidimensional scaling.
\newblock In {\em Recent Developments in Statistics}, 1977.

\bibitem{Ham2005a}
Jihun Ham, Daniel~D. Lee, and Lawrence~K. Saul.
\newblock {Semisupervised alignment of manifolds}.
\newblock In {\em Proceedings of the Annual Conference on Uncertainty in
  Artificial Intelligence, Z. Ghahramani and R. Cowell, Eds}, volume~10, pages
  120--127, 2005.

\bibitem{Hardoon2004}
David~R. Hardoon, Sandor~R. Szedmak, and John~R. Shawe-Taylor.
\newblock Canonical correlation analysis: An overview with application to
  learning methods.
\newblock {\em Neural Computation}, 16:2639--2664, December 2004.

\bibitem{CCA}
Harold Hotelling.
\newblock Relations between two sets of variates.
\newblock {\em Biometrika}, 28(3-4):321--377, 1936.

\bibitem{MaThesis}
Zhiliang Ma.
\newblock {\em Disparate information fusion in the dissimilarity framework}.
\newblock PhD thesis, The Johns Hopkins University, 2010.

\bibitem{measurable_Niemiro1992}
Wojciech Niemiro.
\newblock Asymptotics for {M}-estimators defined by convex minimization.
\newblock {\em The Annals of Statistics}, 20(3):pp. 1514--1533, 1992.

\bibitem{duin2005dissimilarity}
E.~Pekalska and R.P.W. Duin.
\newblock {\em The dissimilarity representation for pattern recognition:
  foundations and applications}.
\newblock Series in machine perception and artificial intelligence. World
  Scientific, River Edge, NJ, 2005.

\bibitem{JOFC}
Carey~E. Priebe, David~J. Marchette, Zhiliang Ma, and Sancar Adali.
\newblock {Manifold matching: Joint optimization of fidelity and
  commensurability}.
\newblock {\em Brazilian Journal of Probability and Statistics},
  27(3):377--400, August 2013.

\bibitem{Raik1972}
E.~Raik.
\newblock On the stochastic programming problem with the probability and
  quantile functionals.
\newblock {\em Izvestia Akademii Nauk Estonskoy SSR. Phys and Math.},
  21(2):142--148, 1972.

\bibitem{Roweis_LLE}
Sam~T. Roweis and Lawrence~K. Saul.
\newblock Nonlinear dimensionality reduction by locally linear embedding.
\newblock {\em Science}, 290(5500):2323--2326, 2000.

\bibitem{Sibson}
Robin Sibson.
\newblock {Studies in the Robustness of Multidimensional Scaling: Procrustes
  Statistics}.
\newblock {\em Journal of the Royal Statistical Society. Series B
  (Methodological)}, 40(2):234--238, 1978.

\bibitem{CMDS}
Warren. Torgerson.
\newblock Multidimensional scaling: I. theory and method.
\newblock {\em Psychometrika}, 17:401--419, 1952.

\bibitem{Wang2008}
C.~Wang and S.~Mahadevan.
\newblock {Manifold alignment using Procrustes analysis}.
\newblock In {\em Proceedings of the 25th international conference on Machine
  learning - ICML '08}, pages 1120--1127, New York, New York, USA, 2008. ACM
  Press.

\bibitem{Zhai2010}
D.~Zhai, B.~Li, H.~Chang, S.~Shan, X.~Chen, and W.~Gao.
\newblock Manifold alignment via corresponding projections.
\newblock In {\em Proceedings of the British Machine Vision Conference}, pages
  3.1--3.11. BMVA Press, 2010.
\newblock doi:10.5244/C.24.3.

\bibitem{ZhuGhodsi}
Mu~Zhu and Ali Ghodsi.
\newblock Automatic dimensionality selection from the scree plot via the use of
  profile likelihood.
\newblock {\em Computational Statistics {\&} Data Analysis}, 51(2):918--930,
  2006.

\end{thebibliography}

\end{document}